\documentclass[journal]{IEEEtran}
\ifCLASSINFOpdf
\else
\fi

\usepackage{hyperref}       
\usepackage{url}            
\usepackage{booktabs}       
\usepackage{amsmath,amssymb,amsfonts}
\usepackage{graphicx}
\usepackage{caption}
\usepackage{subfigure}
\usepackage{array}
\usepackage{algorithm}
\usepackage{algorithmic}
\usepackage{multirow}
\usepackage{indentfirst} 
\usepackage{amsthm}
\usepackage{cite}

\usepackage{color,xcolor}
\usepackage{color}

\graphicspath{ {./images/} }

\newtheorem{lemma}{Lemma}[section]
\newtheorem{theorem}{Theorem}[section]

\begin{document}

    \title{Deep Inertia $L_p$ Half-quadratic Splitting Unrolling Network for Sparse View CT Reconstruction}

    \author{Yu~Guo,~\emph{Graduate Student Member, IEEE}, Caiying Wu,  Yaxin Li,   
        Qiyu~Jin,~\emph{Member, IEEE} and Tieyong Zeng 
       \thanks{This work was supported by National Natural Science Foundation of China (No. 12061052), Natural Science Foundation of Inner Mongolia Autonomous Region (No. 2024LHMS01006, 2024MS01002), Young Talents of Science and Technology in Universities of Inner Mongolia Autonomous Region (No. NJYT22090), Special Funds for Graduate Innovation and Entrepreneurship of Inner Mongolia University (No.~11200-121024), Inner Mongolia University Independent Research Project (No. 2022-ZZ004) and the network information center of Inner Mongolia University. 
       \textit{(Corresponding author: Caiying Wu)}}
    	\thanks{Y. Guo, C. Wu, Y. Li and Q. Jin are with the School of Mathematical Science, Inner Mongolia University, Hohhot 010020, China. (e-mail: yuguomath@aliyun.com; wucaiyingyun@163.com). }
        \thanks{T. Zeng is with Department of Mathematics, The Chinese University of Hong Kong, Satin, Hong Kong.}
      }

    \markboth{Manuscript to IEEE Signal Processing Letters 2024}%
	{Shell \MakeLowercase{\textit{et al.}}: Bare Demo of IEEEtran.cls for IEEE Journals}
	
	\maketitle

	\begin{abstract}
    Sparse view computed tomography (CT) reconstruction poses a challenging ill-posed inverse problem, necessitating effective regularization techniques. In this letter, we employ $L_p$-norm ($0<p<1$) regularization to induce sparsity and introduce inertial steps, leading to the development of the inertial $L_p$-norm half-quadratic splitting algorithm. We rigorously prove the convergence of this algorithm. Furthermore, we leverage deep learning to initialize the conjugate gradient method, resulting in a deep unrolling network with theoretical guarantees. Our extensive numerical experiments demonstrate that our proposed algorithm surpasses existing methods, particularly excelling in fewer scanned views and complex noise conditions.
	\end{abstract}
	

	\section{Introduction}

    Computed tomography (CT) is crucial in modern medicine. However, the ionizing radiation from X-rays has led to interest in low-dose CT. Sparse view CT is an innovative approach that reduces projection data, thereby minimizing radiation and shortening scan times. However, it can cause image degradation, including streak artifacts \cite{BS10}. Mathematically, the CT image reconstruction problem model can be succinctly described as follows:
	\begin{equation}
		f = Au + n,
		\label{pro1}
	\end{equation}
    where $f$ is the measured low quality image, $A$ is the measurement operator, $u$ is the original high quality image, and $n$ is the noise. In fact, in incomplete data CT reconstruction, the matrix $A$ is noninvertible. Hence, the inverse problem (\ref{pro1}) is ill-posed.

    Classical CT reconstruction techniques, such as Algebraic Reconstruction Techniques (ART) \cite{GB70art} and Filtered Backprojection (FBP) \cite{KA02fbp,Yuchuan2005Relation}, struggle to effectively handle sparse view CT. To address this, numerous algorithms have been proposed, with regularization models gaining widespread use. These models typically take the following form:
    \begin{equation}
    	\min_{u} \frac{1}{2} \|Au-f\|^{2} + \lambda R(u),
    \label{pro2}
    \end{equation}
    where $R(u)$ represents the regularization term, and the weighted parameter $\lambda>0$ balances the data fidelity and regularization terms. Various regularization terms, such as BM3D \cite{DF07BM3D,zhao2019ultra,guo2021fast}, total variation \cite{Mahmood2018Adaptive}, sparsity \cite{zhang2018resolution,Bao2019Convolutional}, low rank \cite{kim2014sparse}, and wavelets \cite{DL13JSRmodel,ZL20metainv,Sakhaee2017Joint}, have been employed to tackle this challenge. 
    
    With the advent of deep learning in image processing, an increasing number of sparse view CT reconstructions are leveraging deep learning algorithms to achieve superior performance. These primarily include direct reconstruction using end-to-end network structures \cite{JM17fbpconv,Du2019Visual}, and deep unrolling combined with traditional algorithms \cite{AO18PDnet,ZL20metainv,shi2024mud,Pan2024Iterative,Wu2021DRONE}. The advantage of deep unfolding networks is that they are driven by both model and data \cite{Shi2024Coupling}, meaning that prior knowledge in the domain can assist in learning. Recently, there are new techniques introduced to deep learning, such as attention mechanisms \cite{Wu2023Deep}, and diffusion models \cite{Wu2024Wavelet,Xu2024Stage, Wu2024Multi}, etc. Currently, state-of-the-art sparse-view CT reconstruction results are obtained by score-based model \cite{10570449} and diffusion model \cite{10577271}. While deep learning algorithms offer improved performance, they are dependent on high-quality pairwise data and lack theoretical guarantees, rendering them as black-box algorithms.

    In this letter, we address the sparse view CT reconstruction problem by leveraging the superior sparsity of the $L_p$-norm ($0<p<1$) in conjunction with wavelet transform. To tackle this non-convex and non-continuous problem, we introduce the Inertial $L_p$-norm Half-Quadratic Splitting method (IHQS$_p$), which employs inertial steps to expedite convergence. We establish the convergence of the IHQS$_p$ algorithm. During the CT image reconstruction, the corresponding subproblem is resolved using the conjugate gradient (CG) method. Considering the n-step quadratic convergence of CG, we propose using a network trained via deep learning as the initializer for CG. In the context of the IHQS$_p$ algorithm, the network solely provides the initial value for the CG, thereby not influencing the algorithm’s convergence. From a deep learning perspective, we derive a learning algorithm IHQS$_p$-Net with convergence guarantees, which is deep unrolled by the IHQS$_p$ algorithm. We validate our proposed IHQS$_p$-CG and IHQS$_p$-Net through extensive numerical experiments on two datasets, demonstrating their exceptional performance.

	\section{Inertial HQS$_p$-CG Algorithm and Convergence Analysis}
    \label{sec2}

    \begin{algorithm}[t]
    	\caption{Inertial HQS$_p$-CG algorithm}
    	\label{alg1}
    	\begin{algorithmic}
    		\STATE Initialization: $u^{0}=u_{{\rm FBP}}$, $z^{0}=Wu^{0}, \bar{u}^0=u^{0}, \bar{z}^0=z^{0}, \alpha, \beta, \lambda, \gamma, \epsilon $.
    		\STATE for $k=0:\mathrm{Maxlter}$ do:
    		\STATE \qquad 1. Updata $z^{k+1} = {\rm prox}_{p,\frac{\gamma}{\lambda}}(W\bar{u}^{k})$.
            \STATE \qquad 2. Updata $\bar{z}^{k+1} = z^{k+1} + \alpha(z^{k+1} - \bar{z}^{k})$.
    		\STATE \qquad 3. Updata $u^{k+1} = {\rm CG}(u^{k}, \bar{z}^{k+1}, \gamma)$.
            \STATE \qquad 4. Updata $\bar{u}^{k+1} = u^{k+1} + \beta(u^{k+1} - \bar{u}^{k})$.
            \STATE \qquad 5. Stopping criterion: $\frac{\| \bar{u}^{k+1} - \bar{u}^{k} \|}{\| \bar{u}^{k} \|} \le \epsilon $.
    		\STATE end.
    		\STATE Output: $\bar{u}^{k+1}$
    	\end{algorithmic}
    \end{algorithm}

    \subsection{Inertial HQS$_p$-CG Algorithm}   
    In order to take advantage of the sparsity of the CT image, we use the $L_p$-norm with the wavelet transform to form the regularization term in (\ref{pro2}), i.e.
    \begin{equation}
    	\min_{u} \frac{1}{2} \|Au-f\|^{2} + \lambda \|Wu\|^{p}_{p},
    \label{pro3}
    \end{equation}
    where $0 < p < 1$, $W=(W_{1},W_{2},...,W_{m})$ is a $m$-channel operator with $W_{i}^TW_i=I (1\le i \le m)$. The operator $W$ is chosen as the highpass components of the piecewise linear tight wavelet frame transform  \cite{RS97operator}. Solving (\ref{pro3}) directly is challenging due to the discontinuity in the $L_p$ norm. To solve (\ref{pro3}), we introduce the auxiliary variable $z$, and use the half-quadratic splitting method to obtain the following problem:
    \begin{equation}
		\min_{u,z} \frac{1}{2} \|Au-f\|^{2}_{2} + \lambda \|z\|^{p}_{p} + \frac{1}{2} \sum_{i=1}^{m} \gamma_{i} \|W_{i}u-z_{i}\|^{2}_{2}, 
	\label{eq6}
	\end{equation}
    where $\lambda > 0$ and $\gamma = (\gamma_{1},\gamma_{2},...,\gamma_{m})$ with $\gamma_{i}>0$ $(i=1,2,...,m)$ are weight parameters. The above problem can be split into $z$-subproblem and $u$-subproblem solved in alternating iterations. To speed up the convergence, inspired by the inertia algorithm \cite{alvarez2001inertial,ochs2014ipiano}, we introduce inertia steps in solving the two subproblems, i.e.,
    \begin{align}
        z^{k+1} &= \arg\min_{z} \lambda \|z\|^{p}_{p} + \frac{1}{2}\sum_{i=1}^{m}\gamma_{i}\|W_{i}\bar{u}^{k}-z_{i}\|^{2}_{2}, \label{eq7} \\
		\bar{z}^{k+1} &= z^{k+1} + \alpha(z^{k+1} - \bar{z}^{k}), \label{eq8} \\
		u^{k+1} &= \arg\min_{u} \frac{1}{2}\|Au-f\|^{2}_{2} +\frac{1}{2} \sum_{i=1}^{m} \gamma_{i} \|W_{i}u-\bar{z}_{i}^{k+1}\|^{2}_{2}, \label{eq9} \\
	  \bar{u}^{k+1} &= u^{k+1} + \beta(u^{k+1} - \bar{u}^{k}),  \label{eq10}
    \end{align} 
    where (\ref{eq7}) and (\ref{eq9}) are subproblem solving steps, (\ref{eq8}) and (\ref{eq10}) are inertia steps, and $\alpha \in [0,1)$ and $\beta \in [0,\frac{\sqrt{5}-1}{2})$ are inertia weight parameters.

     The $z^{k+1}$-update step (\ref{eq7}) can be solved by the proximal operator ${\rm prox}_{p,\eta}$, which is defined as
    \begin{equation}
    	{\rm prox}_{p,\eta}(t)=\mathop{\arg\min}\limits_{x}{\|x\|_{p}^{p}+\frac{\eta}{2}\|x-t\|^{2}}.
    \end{equation}
    where $0<p<1$, it can be computed as \cite{MS12l_q}
	\begin{equation}
	{\rm prox}_{p,\eta}(t)=\left\{
	\begin{array}{ll}
		0, & \left|t\right|<\tau\\
		\{0,{\rm sign}(t)\rho \}, & \left|t\right|=\tau\\ 
		{\rm sign}(t)g, & \left|t\right|>\tau
	\end{array}
	\right.
	\end{equation}
	where $\rho = ( 2(1-p) / \eta)^\frac{1}{2-p}$, $\tau=\rho + p\rho^{p-1} / \eta$, $g$ is the solution of $h(g)=pg^{p-1}+\eta g-\eta|t|=0$ over the region $(\rho,|t|)$. Since $h(g)$ is convex, when $|t|>\tau$, $g$ can be efficiently solved using Newton’s method.

    For the $u^{k+1}$-update step (\ref{eq9}), according to the Euler-Lagrange equations, we have: 
    \begin{align*}
        u^{k+1}=(A^{T}A+\sum_{i=1}^{m}\gamma_{i} W_{i}^{T}W_{i})^{-1}(A^{T}f+\sum_{i=1}^{m}\gamma_{i}W_{i}^{T}\bar{z}_{i}^{k}).
    \label{eq11}
    \end{align*}
    To avoid the inverse, this equation can be solved using the conjugate gradient method. We call our algorithm as Inertial HQS$_p$-CG (IHQS$_p$-CG) for short and summarized in Algorithm \ref{alg1}.
	Next, we establish the global convergence of the IHQS$_p$-CG algorithm.

    \begin{figure}[t]
    \begin{center}
    \footnotesize
    \addtolength{\tabcolsep}{-4pt}
    \renewcommand\arraystretch{0.5}
    {\fontsize{10pt}{\baselineskip}\selectfont
    \begin{tabular}{cc}
        \includegraphics[width=.22\textwidth]{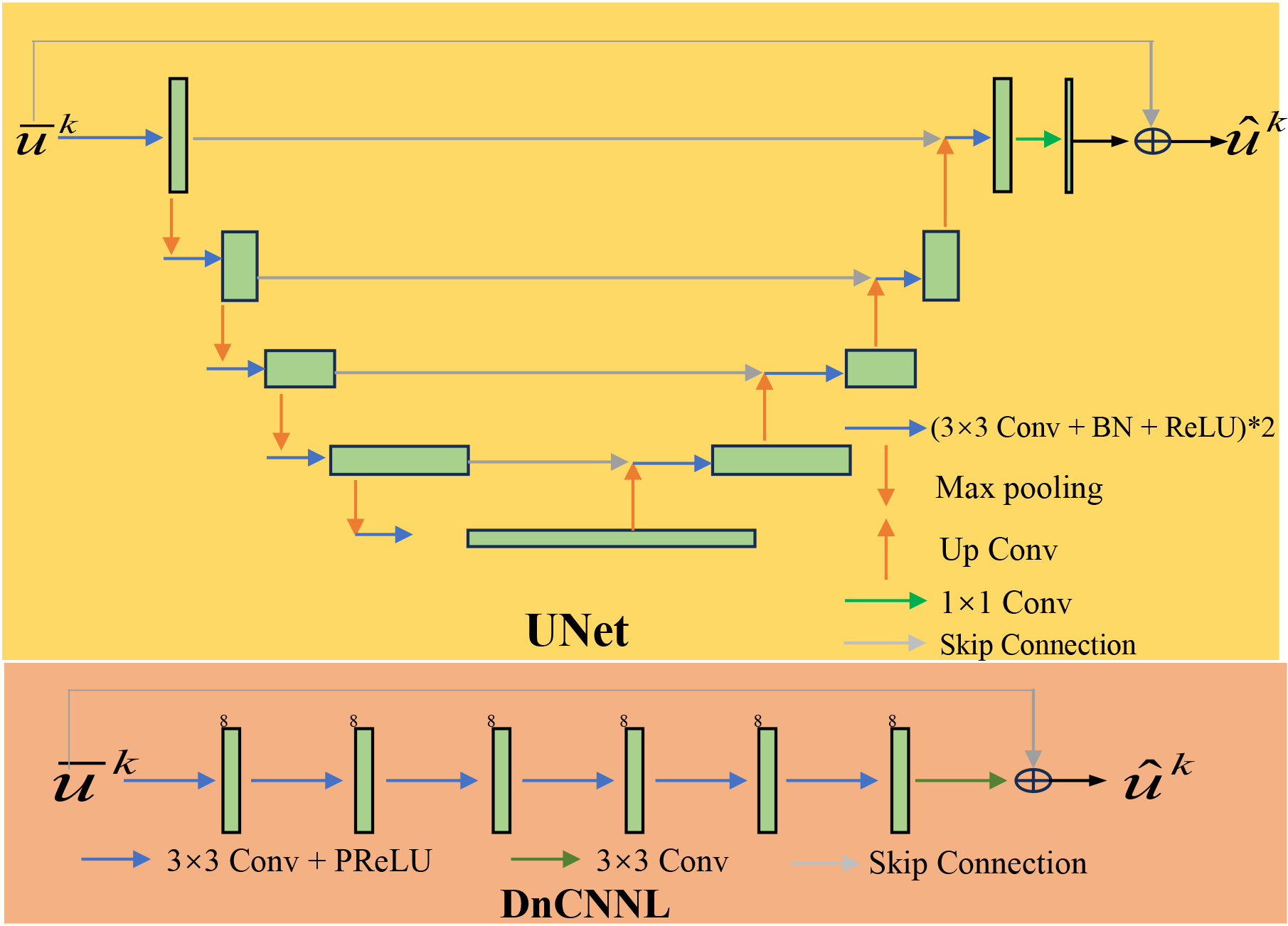} \\
        \includegraphics[width=.4\textwidth]{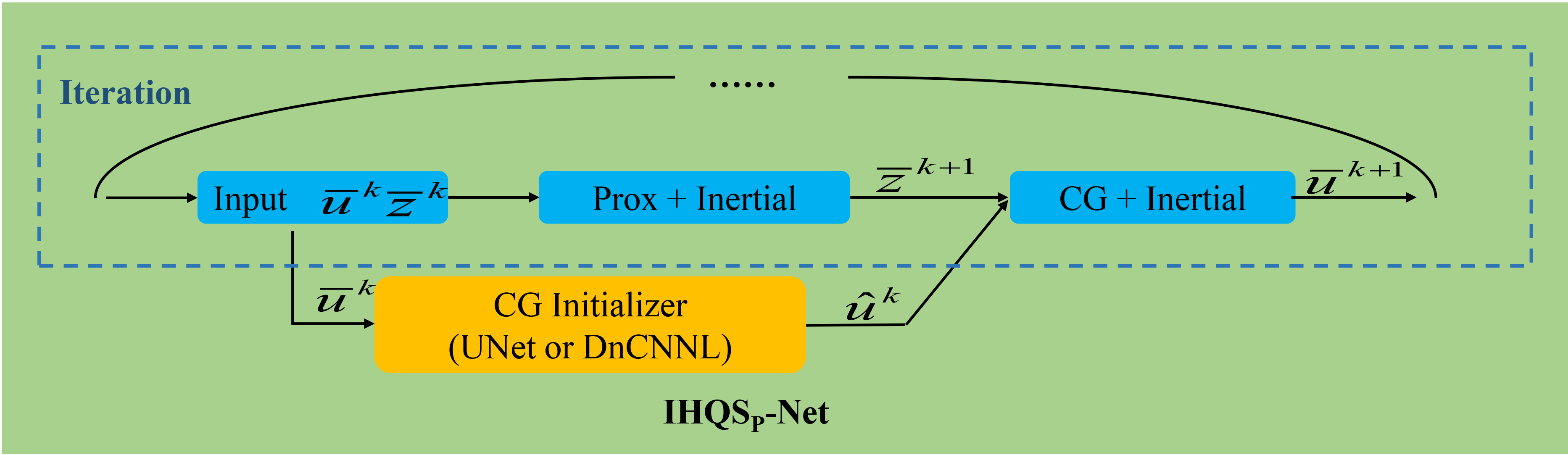}  \\
    \end{tabular}
    }
    \end{center}
    \caption{IHQS$_p$-Net algorithm framework diagram.}
    \label{net}
    \end{figure}
    
	\subsection{Convergence analysis}
	In this section, we provide a convergence analysis of Algorithm \ref{alg1}.

	\begin{lemma} 
    \label{lemma1}
	Suppose that the sequences $\left\{v^{k}\right\}$ and $\left\{\bar{v}^{k}\right\}$ generated via Algorithm \ref{alg1}, $0<\beta<\frac{\sqrt{5}-1}{2}$, then the sets $\left\{v^{k}\right\}$ and $\left\{\bar{v}^{k}\right\}$ are bounded.
    \end{lemma}

    \begin{theorem}
    Let $\left\{v^{k}\right\}$ and $\left\{\bar{v}^{k}\right\}$ be the sequences generated by our algorithm. Then any cluster point of $\left\{v^{k}\right\}$ is the global minimum point of $L$. 
    \end{theorem}

    Details of the proof can be found in the supplementary material.

    \begin{table*}[t]
    \caption{Comparison of results between different algorithms on the AAPM and Covid-19 datasets (PSNR/SSIM). The best value of the traditional algorithm is marked with \textbf{bold}. The best value for deep learning is marked with {\color{red}red}. }
    \addtolength{\tabcolsep}{-4pt}
    \renewcommand\arraystretch{0.85}
    \scriptsize
    \begin{center}
    {\begin{tabular}{|l|c|c|c|c|c|c|c|c|c|c|c|} 
        \hline
        Data & \multicolumn{11}{c|}{AAPM Dataset}  \\ \hline
        Method & \multicolumn{6}{c|}{Traditional algorithm }  & \multicolumn{5}{c|}{Deep learning algorithm}  \\ \hline
        Method & FBP & BM3D & HQS-CG & HQS$_p$-CG & PWLS-CSCGR & IHQS$_p$-CG & PDNet & FBPNet & MetaInvNet & IHQS$_p$-Net* & IHQS$_p$-Net \\ \hline

        Noise & \multicolumn{10}{c|}{$\sigma=0.3$}\\\hline

        180 & 23.64/0.5136 & 24.82/0.6839 & 31.65/0.8047 & 32.10/0.8131 & 26.72/24.7975 & \textbf{32.16}/\textbf{0.8138} & 30.06/0.8124 & 31.73/0.8138 & 33.01/0.8348 & 32.70/0.8297 & {\color{red} 33.18}/{\color{red} 0.8415} \\ 
        120 & 20.69/0.4104 & 23.31/0.6454 & 30.11/0.7659 & 30.72/0.7841 & 26.97/23.9243 & \textbf{30.86}/\textbf{0.7865} & 28.89/0.7813 & 30.36/0.7811 & {\color{red} 32.17}/0.8103 & 31.45/0.7999 & 32.08/{\color{red} 0.8148} \\ 
        90 & 20.11/0.3494 & 23.55/0.6143 & 28.83/0.7324 & 29.59/0.7602 & 25.76/22.8060 & \textbf{29.81}/\textbf{0.7647} & 28.26/0.7633 & 29.60/0.7674 & 31.56/0.7939 & 30.47/0.7777 & {\color{red} 31.66}/{\color{red} 0.8032} \\ 
        60 & 18.21/0.2762 & 22.34/0.5580 & 26.84/0.6809 & 27.79/0.7226 & 23.94/21.0367 & \textbf{28.06}/\textbf{0.7298} & 26.61/0.7298 & 28.06/0.7397 & {\color{red} 30.38}/0.7627 & 28.89/0.7399 & 30.24/{\color{red} 0.7671} \\ 
           \hline

        Noise & \multicolumn{11}{c|}{$\sigma=0.3$ and $I_0=5\times 10^{5}$}\\\hline

        180 & 22.77/0.4397 & 24.82/0.6830 & 30.57/0.7641 & 30.91/0.7845 & 25.98/22.0973 & \textbf{30.92}/\textbf{0.7858} 
            & 29.54/0.7922 & 31.38/0.7990 & 32.02/0.8080 & 31.53/0.8005 & {\color{red} 32.25}/{\color{red} 0.8157} \\ 
        120 & 20.04/0.3494 & 23.31/0.6443 & 29.08/0.7179 & 29.81/0.7595 & 22.45/14.5781 & \textbf{29.90}/\textbf{0.7619} & 28.17/0.7633 & 30.15/0.7750 & 31.33/0.7929 & 30.48/0.7743 & {\color{red} 31.40}/{\color{red} 0.7961} \\ 
        90 & 19.37/0.2968 & 23.55/0.6132 & 27.91/0.6832 & 28.78/0.7369 & 22.94/15.7751 & \textbf{28.97}/\textbf{0.7430} & 27.16/0.7363 & 29.19/0.7570 & 30.76/0.7674 & 29.71/0.7516 & {\color{red} 30.98}/{\color{red} 0.7815} \\ 
        60 & 17.55/0.2341 & 22.37/0.5575 & 26.18/0.6359 & 27.09/0.6961 & 22.63/16.3278 & \textbf{27.30}/\textbf{0.7049} & 25.51/0.7077 & 28.13/0.7332 & 29.66/0.7424 & 28.62/0.7213 & {\color{red} 30.08}/{\color{red} 0.7597} \\ 
        \hline

        Data & \multicolumn{11}{c|}{Covid-19 Dataset}  \\ \hline
        Method & FBP & BM3D & HQS-CG & HQS$_p$-CG & PWLS-CSCGR & IHQS$_p$-CG & PDNet & FBPNet & MetaInvNet & IHQS$_p$-Net* & IHQS$_p$-Net \\ \hline

        Noise & \multicolumn{10}{c|}{$\sigma=0.3$}\\\hline

        180 & 23.91/0.5361 & 25.46/0.7304 & 33.32/0.8393 & 33.93/0.8496 & 28.50/20.5285 & \textbf{34.05}/\textbf{0.8515} & 30.71/0.8413 & 33.21/0.8512 & 34.76/0.8638 & 34.43/0.8584 & {\color{red} 35.03}/{\color{red} 0.8704} \\ 
        120 & 21.07/0.4303 & 24.19/0.6906 & 31.56/0.8041 & 32.29/0.8241 & 28.25/19.7033 & \textbf{32.64}/\textbf{0.8285} & 29.61/0.8140 & 31.66/0.8196 & {\color{red} 34.06}/{\color{red} 0.8523} & 32.95/0.8290 & 33.83/0.8489 \\ 
        90 & 20.50/0.3665 & 24.11/0.6585 & 29.91/0.7707 & 30.83/0.8015 & 27.01/19.1825 & \textbf{31.29}/\textbf{0.8078} & 29.12/0.7980 & 30.81/0.8074 & 33.25/0.8380 & 31.93/0.8120 & {\color{red} 33.52}/{\color{red} 0.8432} \\ 
        60 & 18.47/0.2897 & 22.26/0.5966 & 27.30/0.7179 & 28.36/0.7621 & 24.98/17.8874 & \textbf{28.94}/\textbf{0.7731} & 27.34/0.7651 & 28.36/0.7740 & {\color{red} 31.88}/{\color{red} 0.8150} & 30.12/0.7729 & 31.81/0.8147 \\  
        \hline
        
        Noise & \multicolumn{11}{c|}{$\sigma=0.3$ and $I_0=5\times 10^{5}$}\\ \hline
        
        180 & 22.96/0.4426 & 25.45/0.7294 & 31.81/0.7919 & 32.30/0.8163 & 27.44/17.9762 & \textbf{32.34}/\textbf{0.8188} & 30.13/0.8199 & 32.84/0.8367 & 33.69/0.8452 & 32.75/0.8250 & {\color{red} 33.98}/{\color{red} 0.8482} \\ 
        120 & 20.32/0.3515 & 24.18/0.6893 & 30.17/0.7480 & 31.10/0.7967 & 22.97/11.6984 & \textbf{31.33}/\textbf{0.8009} & 29.35/0.7979 & 31.45/0.8154 & 32.87/0.8302 & 31.65/0.7993 & {\color{red} 33.06}/{\color{red} 0.8326} \\ 
        90 & 19.65/0.2980 & 24.10/0.6572 & 28.78/0.7147 & 29.81/0.7748 & 23.63/13.0201 & \textbf{30.21}/\textbf{0.7832} & 27.78/0.7682 & 30.45/0.7979 & 32.18/0.8150 & 30.79/0.7787 & {\color{red} 32.63}/{\color{red} 0.8239} \\ 
        60 & 17.71/0.2340 & 22.27/0.5956 & 26.58/0.6674 & 27.60/0.7328 & 23.41/13.8117 & \textbf{28.05}/\textbf{0.7444} & 26.25/0.7414 & 28.80/0.7729 & 30.80/0.7897 & 29.39/0.7493 & {\color{red} 31.62}/{\color{red} 0.8048} \\  
        \hline

        Cpu (s) & 0.99 & 3.45 &  18.63 & 34.72 & --      & 26.24 & 2.03 & 5.18 &  6.76 & 2.08 & 6.98 \\
        Gpu (s) & --   & --   & 8.86   & 17.63 & 1530.59 & 15.91 & 0.62 & 1.35 &  1.92 & 1.98 & 2.29 \\ 
        \hline
    
    \end{tabular}}
    \end{center}
    \label{Covid}
    \end{table*}

    \section{Conjugate gradient initialization and deep unrolling networks}
    \label{sec3}

    The algorithm, as per Algorithm \ref{alg1}, is divided into two subproblems: z-subproblem for denoising and u-subproblem for CT image reconstruction. The u-subproblem solution uses the conjugate gradient method, known for its super-linear convergence rate. The convergence speed is tied to the spectral distribution of the coefficient matrix, with faster convergence when eigenvalues are concentrated and the condition number is smaller \cite{Smyl2021An,Savanier2022Unmatched}. A well-defined initial value can lead to finite iterations for convergence. To facilitate this, we introduce a deep learning network to adjust the input $\bar{u}^{k}$ at each iteration:
    \begin{equation*}
    	\hat{u}^{k} = \bar{u}^{k} + {\rm Net}(\bar{u}^{k}; \theta),
    	\label{3.1}
    \end{equation*}
    where $\theta$ represents the learnable network parameter, and $\hat{u}^{k}$ is the corrected result. The algorithm’s structure can incorporate any network structure, with UNet and lightweight DnCNN as examples, referred to as IHQS$_p$-Net and IHQS$_p$-Net*, respectively. A detailed description is shown in Fig.   \ref{net}. The parameter $\theta$ is trained through a loss function:
    \begin{equation*}
    	\mathcal{L}(\theta) = \frac{1}{K} \sum_{k=1}^{K} (\mu_{1}\mathcal{L}_{2}(\hat{u}^{k}, f) + \mu_{2}\mathcal{L}_{{\rm SSIM}}(\hat{u}^{k}, f)),
    	\label{3.2}
    \end{equation*}
    where $\mu_1$ and $\mu_2$ are weight factors, we set $\mu_1 = \mu_2 = 1$. $\mathcal{L}_{2}$ is the $L_{2}$ loss, i.e. $\mathcal{L}_{2}(x,y) = \sum \|x-y\|^2$. $\mathcal{L}_{SSIM}$ is the structural similarity loss, i.e. $\mathcal{L}_{SSIM}(x,y) = 1 - \mathrm{SSIM}(x,y)$ \cite{04ssim}.

    From an algorithmic view, the deep network is embedded in the splitting algorithm, serving as an initializer for CG solving without affecting convergence. From a deep learning view, the splitting algorithm framework is a deep unrolling network with theoretical guarantees. The integration of both enhances performance while ensuring theoretical robustness.

	\section{Numerical experiments}
    \label{sec4}
    
	\subsection{Datasets and experimental settings}

   The deep learning algorithm used training and validation data from the “2016 NIH-AAPM-Mayo Clinic Low Dose CT Grand Challenge” dataset \cite{M16aapm}. The training set had 1596 slices from five patients, and the validation set had 287 slices from one patient. Experiments were uniformly performed using a fan beam geometry with 800 detector elements. The test data included 38 slices from new patients in the AAPM dataset and 30 slices from the Covid-19 dataset \cite{SW20COVID} to assess the algorithm’s generalization and robustness.

    Our proposed algorithms were compared with eight others, including five traditional (FBP \cite{KA02fbp}, BM3D \cite{DF07BM3D}, HQS-CG \cite{ZL20metainv}, combining compressed sensing and sparse convolutional coding for PWLS-CSCGR \cite{Bao2019Convolutional}, HQS$_p$-CG without inertial $L_p$ norms) and three deep learning algorithms (PDNet \cite{AO18PDnet}, FBPNet \cite{JM17fbpconv}, MetaInvNet \cite{ZL20metainv}). For our lightweight IHQS$_p$-Net* and IHQS$_p$-Net, we set K = 6 and limited the CG algorithm iterations to 7. Data were degraded with different angles (60, 90, 120, 180) and two noise levels (Gaussian noise $\sigma=0.3$, mixed Gaussian and Poisson noise with intensity $I_0=5\times 10^{5}$).

	Our models were implemented using the PyTorch framework. The experiments were conducted on PyTorch 1.3.1 backend with an NVIDIA Tesla V100S GPU, 8-core CPU and 32GB RAM. For IHQS$_p$-CG, the hyperparameters were set to $p=0.7, \alpha=0.5$, $\beta=0.6$, $\lambda \in (0.0004,0.0008)$, $\gamma \in (0.1,0.4)$. For IHQS$_p$-Net, the hyperparameters were set to $p=0.7, \alpha=0.3$, $\beta=0.1$, $\lambda \in (0.001,0.004)$, $\gamma \in (0.01, 0.06)$.

    \begin{figure*}[!t]
        \begin{center}
        \scriptsize
        \addtolength{\tabcolsep}{-5pt}
        \renewcommand\arraystretch{0.5}
        {\fontsize{9pt}{\baselineskip}\selectfont
        \begin{tabular}{cccccc}
            \includegraphics[width=0.14\linewidth]{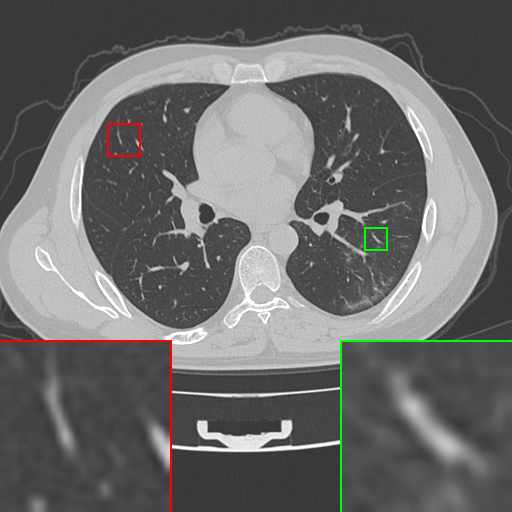} &
            \includegraphics[width=0.14\linewidth]{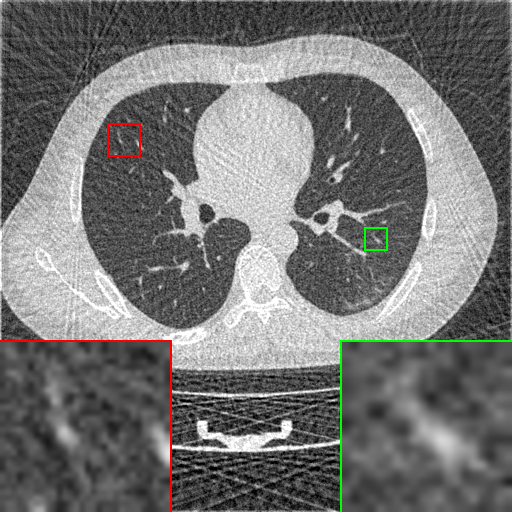} &
            \includegraphics[width=0.14\linewidth]{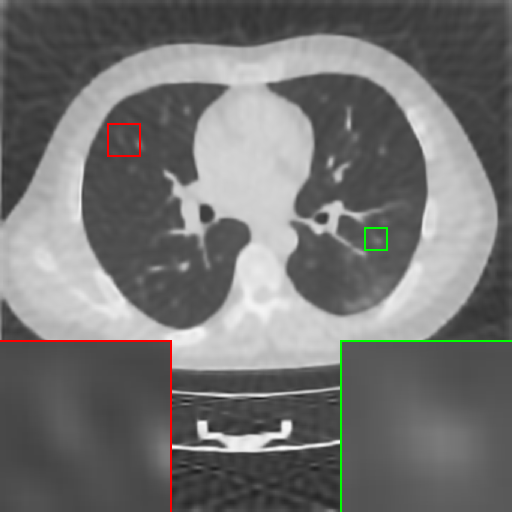} &
            \includegraphics[width=0.14\linewidth]{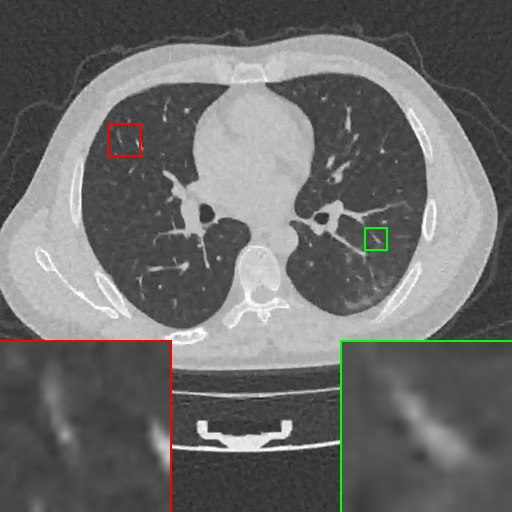} &
            \includegraphics[width=0.14\linewidth]{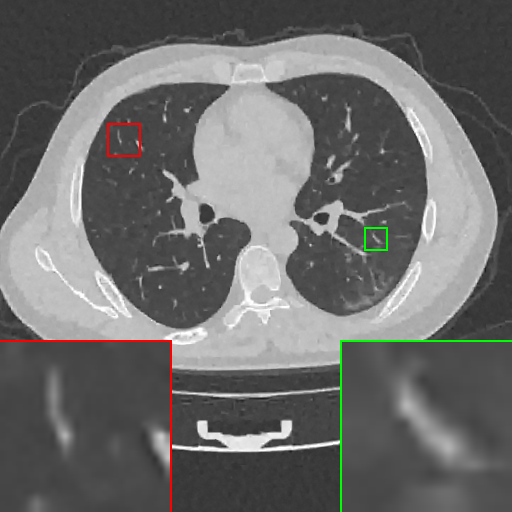} & 
            \includegraphics[width=0.14\linewidth]{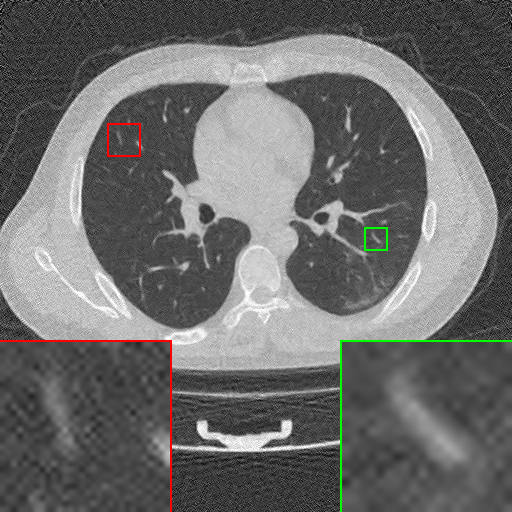} \\
            
            (a) Ground truth & (b) FBP & (c) BM3D & (d) HQS-CG & (e) HQS$_p$-CG & (f) PWLS-CSCGR  \\
    
            \includegraphics[width=0.14\linewidth]{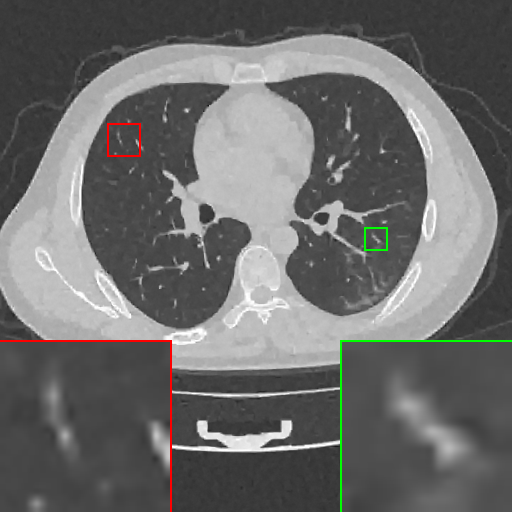} &
            \includegraphics[width=0.14\linewidth]{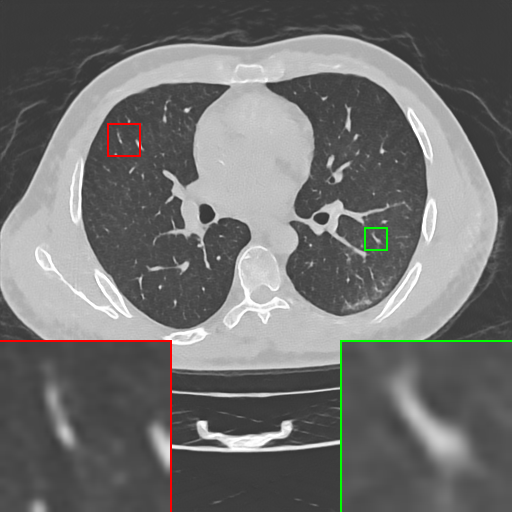} &
            \includegraphics[width=0.14\linewidth]{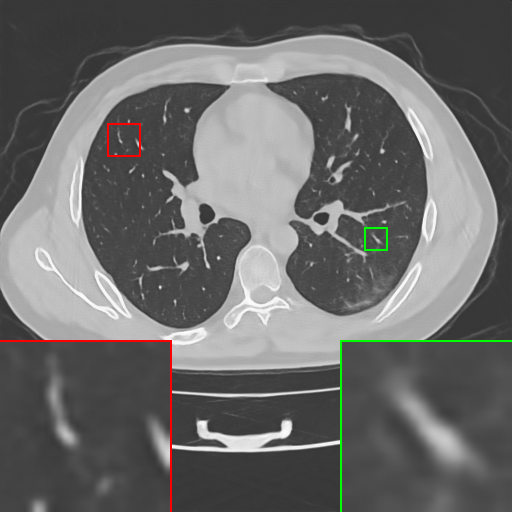} &
            \includegraphics[width=0.14\linewidth]{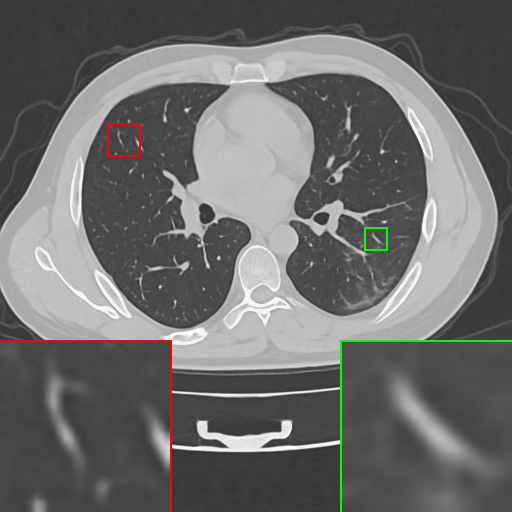} &
            \includegraphics[width=0.14\linewidth]{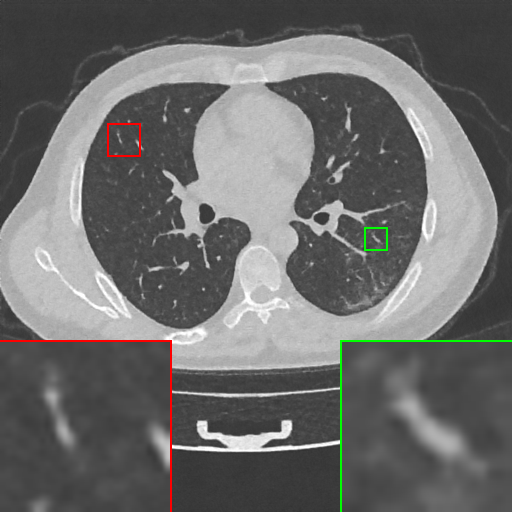} &
            \includegraphics[width=0.14\linewidth]{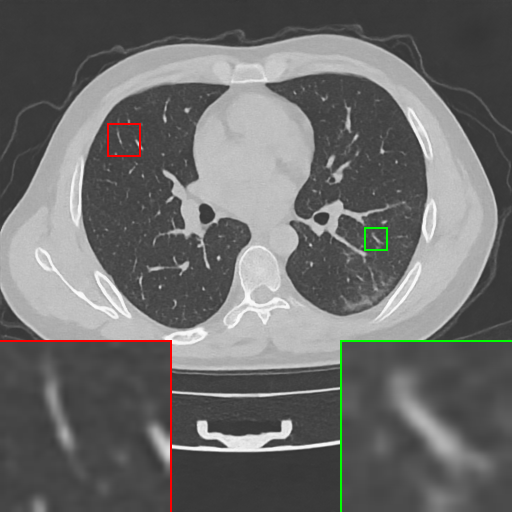} \\
            
           (g) IHQS$_p$-CG & (h) PDNet & (i) FBPNet & (j) MetaInvNet & (k) IHQS$_p$-Net* & (l) IHQS$_p$-Net \\
        \end{tabular}
        }
        \end{center}
        \caption{120 sparse views CT image reconstruction results and magnified ROIs from Covid-19. The sinogram is corrupted by $\sigma=0.3$ and $I_0=5\times 10^{5}$. 
        }
        \label{reconstruction1}
    \end{figure*}

    \begin{figure}[!t]
        \begin{center}
        \scriptsize
        \addtolength{\tabcolsep}{-5pt}
        \renewcommand\arraystretch{0.5}
        {\fontsize{10pt}{\baselineskip}\selectfont
        \begin{tabular}{cc}
        \includegraphics[width=0.42\linewidth, trim={23 1 40 20}, clip]{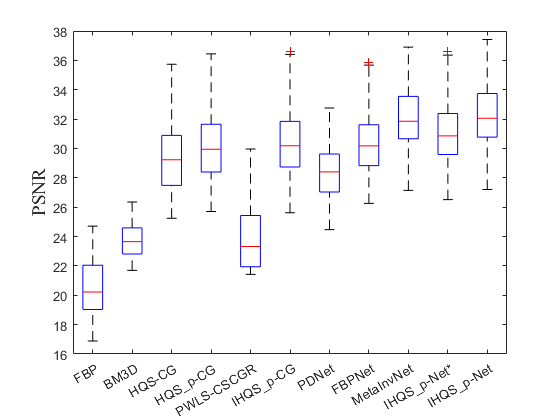} &
        \includegraphics[width=0.42\linewidth, trim={23 1 40 20}, clip]{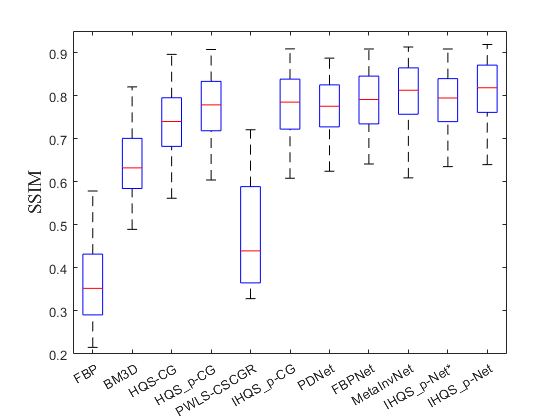} \\
        \end{tabular}
        }
        \end{center}
        \caption{Robustness analysis of all compared algorithms.}
        \label{box}
    \end{figure}

	\subsection{Comparison of quantitative and qualitative results}

    Evaluation metrics used are peak signal-to-noise ratio (PSNR) \cite{10psnr} and structural similarity (SSIM) \cite{04ssim}. For a comprehensive evaluation of the algorithm performance, we added a comparison of the mean absolute error (MAE) and the root mean square error (RMSE) in the supplementary material. Table \ref{Covid} shows that our proposed IHQS$_p$-CG algorithm outperforms other traditional algorithms, even surpassing the deep learning algorithms PDNet and FBPNet on both AAPM and Covid-19 datasets. It shows greater robustness to noise and missing range than the HQS-CG algorithm, with an average improvement of more than 0.5 dB for 180-degree and 120-degree observations, and over 1 dB for 90-degree and 60-degree observations.

    Our proposed IHQS$_p$-Net outperforms other deep learning methods, especially in handling mixed noises. Our lightweight IHQS$_p$-Net* significantly improves upon the traditional IHQS$_p$-CG, indicating a lightweight CG initializer can enhance performance and convergence speed. Compared to FBPNet, IHQS$_p$-Net shows an average improvement of 2 dB, highlighting the benefits of integrating traditional algorithms with theoretical guarantees into deep networks. Meanwhile, Table \ref{Covid} gives the running time for each algorithm to reconstruct a 512$\times$512 observation of 90 degrees.

    Fig.  \ref{reconstruction1} compares reconstructed images. FBP, BM3D, and HQS-CG show artifacts and missing structures due to lack of observation angle. Our IHQS$_p$-CG yields smoother results. Deep learning algorithms like PDNet and FBPNet produce smooth but over-smoothed results, blurring structural information. Both MetaInvNet and our IHQS$_p$-Net produce pleasing results, but IHQS$_p$-Net captures more details and avoids pseudo-connections seen in MetaInvNet. Notably, IHQS$_p$-Net* visually outperforms IHQS$_p$-CG, highlighting the importance of the CG initializer. The robustness analysis of all the algorithms is given in Fig. 3. By statistically analyzing the reconstructed images for all noise cases, observation angle cases and datasets, it can be found that the proposed IHQS$_p$-CG and IHQS$_p$-Net are robust to noise, observation angle, and data.

    \begin{figure}[!t]
    \begin{center}
    \footnotesize
    \addtolength{\tabcolsep}{-5pt}
    \renewcommand\arraystretch{0.5}
    {\fontsize{9pt}{\baselineskip}\selectfont
    \begin{tabular}{cc}
    \includegraphics[width=0.44\linewidth, trim={17 5 3 20}, clip]{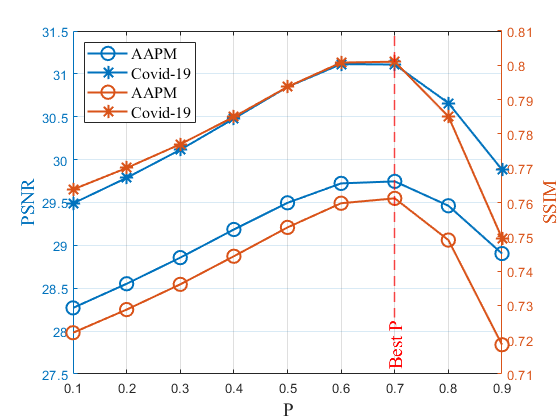} &
    \includegraphics[width=0.44\linewidth, trim={0 1 7 20}, clip]{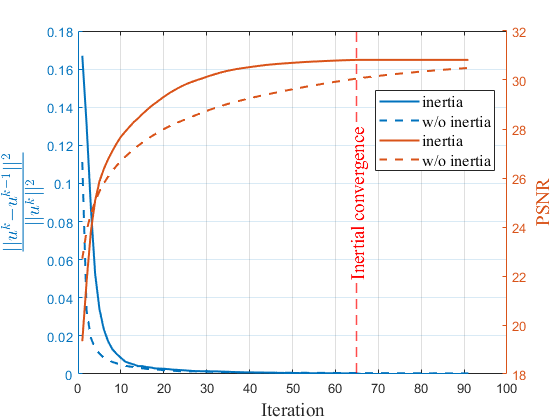} \\
    (a) $p$ & (b) Convergence
    \end{tabular}
    }
    \end{center}
    \caption{(a) Regarding the ablation experiment of $p$-value. (b) The impact of inertia step on algorithm convergence.}
    \label{p}
    \end{figure}

	\subsection{Ablation experiment}

    There is a key parameter $p$ in IHQS$_p$ which controls the sparsity of the regular term. Different values of $p$ affect the reconstruction results of CT images \cite{zhang2018resolution}. Therefore, we need to discuss the different $p$-value effects. Fig. \ref{p}(a) shows the results of different $p$ values on two datasets, AAPM and Covid-19. From the figure, we can find that between 1 and 0.7, the reconstruction effect is enhanced as $p$ decreases. This reflects the advantage of the $L_p$ paradigm over $L_1$. However, lower$ p$ is not better. As can be seen in the figure, for IHQS$_p$, the optimal $p = 0.7$.

    One of the contributions of IHQS$_p$ is the introduction of inertial steps to accelerate the convergence of HQS. HQS$_p$-CG, in Table \ref{Covid}, shows the numerical results without inertia step. It can be found that the introduction of the inertia step can effectively improve the algorithm's performance, and the improvement becomes more obvious as the observation angle becomes sparser.  Fig. \ref{p}(b) gives a comparison between IHQS$_p$ and HQS$_p$-CG for iterative reconstruction when the 90-degree observation is corrupted by mixed Gaussian Poisson noise. PSNR and error $\frac{|| u^k - u^{k-1}||^2}{|| u^k ||^2}$ are used as metrics. From Fig. \ref{p}(b), IHQS$_p$ stops iterating at step 65, while HQS$_p$ without inertial steps iterates until step 90. Meanwhile, HQS$_p$ with more iterations does not outperform IHQS$_p$ in PSNR. This reflects the importance of introducing inertial steps.

    \section{Conclusion}
    \label{sec5}

    We propose an inertial $L_p$-norm half-quadratic splitting algorithm for sparse view CT image reconstruction and establish its convergence. Building on IHQS$_p$, we introduce IHQS$_p$-Net, a deep unrolling network that comes with theoretical guarantees. Both quantitative and qualitative comparisons with other methods demonstrate that our proposed IHQS$_p$-CG and IHQS$_p$-Net outperform in terms of performance and robustness, particularly excelling in scenarios with fewer scanned views and complex noise situations.

\newpage

\bibliographystyle{IEEEtran}
\IEEEtriggeratref{19}
\bibliography{IEEEabrv,reference}

\newpage
\section{Convergence analysis}

Define $v^{k}=(z^{k},u^{k}),  \bar{v}^{k}=(\bar{z}^{k},\bar{u}^{k})$,
\begin{equation*}
	\begin{aligned}
		L(v^{k})=\frac{1}{2}\|Au^{k}-f\|^{2}_{2}+\lambda\|z^{k}\|^{p}_{p}+\frac{1}{2}\sum_{i=1}^{m}\gamma_{i}\|W_{i}u^{k}-z_{i}^{k}\|^{2}_{2}.
	\end{aligned}
\end{equation*}

\begin{lemma} 
	\label{lemma1}
	Suppose that the sequences $\left\{v^{k}\right\}$ and $\left\{\bar{v}^{k}\right\}$ generated via Algorithm 1, $0<\beta<\frac{\sqrt{5}-1}{2}$, then the sets $\left\{v^{k}\right\}$ and $\left\{\bar{v}^{k}\right\}$ are bounded.
\end{lemma}

\begin{proof}[Proof]
	In our problem, $u^{k}$ represents the gray value of an image, i.e., $0\leq u^{k}_{i,j} \leq255$, therefore $\left\{u^{k}\right\}$ is bounded. From the $u^{k+1}$-update step (7), we have 
	\begin{equation*}
		u^{k+1}=(A^{T}A+\sum_{i=1}^{m}\gamma_{i} W_{i}^{T}W_{i})^{-1}(A^{T}f+\sum_{i=1}^{m}\gamma_{i}W_{i}^{T}\bar{z}_{i}^{k}),
	\end{equation*}
	which together with $W_{i}^{T}W_{i}=I$ and $W_{i}$ is the given wavelet transform operator, yields,
	\begin{equation*}
		\sum_{i=1}^{m}\gamma_{i}W_{i}^{T}\bar{z}_{i}^{k} = (A^{T}A+\sum_{i=1}^{m}\gamma_{i}I)u^{k+1} - A^{T}f.
	\end{equation*}
	From $\left\{u^{k}\right\}$ is bounded, it is easy to see that $\left\{\bar{z}^{k}\right\}$ is bounded.
	
	On the other hand, from (8), we find 
	\begin{equation}
		\begin{aligned}
			\bar{u}^{k+1} &= u^{k+1} + \beta(u^{k+1} - \bar{u}^k) = (1+\beta)u^{k+1} - \beta\bar{u}^k \\
			&= (1+\beta)u^{k+1} - \beta(u^{k} + \beta(u^{k} - \bar{u}^{k-1})) \\
			&= (1+\beta)u^{k+1} - \beta(1+\beta)u^{k} + \beta^2(1+\beta)u^{k-1} - \beta^3\bar{u}^{k-2} \\
			&=  \cdots  \cdots \\
			&= (1+\beta)u^{k+1} - \beta(1+\beta)u^{k}  +  \cdots + (-1)^{k}\beta^{k}(1+\beta)u^{1} \\
			&~~~+ (-1)^{k+1}\beta^{k+1}u^{0}.
		\end{aligned}
		\label{eq12}
	\end{equation}
	Since $\left\{u^{k}\right\}$ is bounded, there exists a constant $N>0$ such that $\| u^{k} \| \leq N$, $\forall k \geq 0$. If $0<\beta<\frac{\sqrt{5}-1}{2}$ holds, then $\beta(1+\beta)<1$. It follows from (\ref{eq12}) that
	\begin{equation*}
		\begin{aligned}
			\| \bar{u}^{k+1} \| &\leq (1+\beta)\| u^{k+1} \| + \beta(1+\beta)\|u^{k}\| + \cdots \\
			&~~~+ \beta^{k}(1+\beta)\|u^{1}\| + \beta^{k+1}\|u^{0}\| \\
			&< (1+\beta)N + N + \beta N +\beta^{2}N + \cdots + \beta^{k-1}N + \beta^{k}N \\
			&< (1+\beta)N + N(1+ \beta  +\beta^{2} + \cdots + \beta^{k-1} + \beta^{k}) \\
			&= (1+\beta)N + N\sum_{i=1}^{k}\beta^{i} \\
			&< (1+\beta)N + N\sum_{i=1}^{\infty}\beta^{i} = (1+\beta)N + \frac{N}{1-\beta} \\
			&=\frac{2-\beta^2}{1-\beta}N. 
		\end{aligned}
		\label{eq13}
	\end{equation*}
	which consequently results in the boundedness of sequence $\left\{u^{k}\right\}$. By (6), we have $(1+\alpha)z^{k+1} = \bar{z}^{k+1} + \alpha\bar{z}^{k}$. Thus, $\left\{z^{k}\right\}$ is bounded because $\left\{\bar{z}^{k}\right\}$ is bounded. Therefore, for $k \geq 0$, we get $\left\{v^{k}\right\}$ and $\left\{\bar{v}^{k}\right\}$ are bounded. 
\end{proof}

\begin{table*}[!t]
	\caption{Comparison of results between different algorithms on the AAPM and Covid-19 datasets (MAE/RMSE). The best value of the traditional algorithm is marked with \textbf{bold}. The best value for deep learning is marked with {\color{red}red}. }
	\addtolength{\tabcolsep}{-4pt}
	\renewcommand\arraystretch{1}
	\begin{center}
		{\begin{tabular}{|l|c|c|c|c|c|c|c|c|c|c|c|} 
				\hline
				Data & \multicolumn{11}{c|}{AAPM Dataset}  \\ \hline
				Method & \multicolumn{6}{c|}{Traditional algorithm }  & \multicolumn{5}{c|}{Deep learning algorithm}  \\ \hline
				Method & FBP & BM3D & HQS-CG & HQS$_p$-CG & PWLS-CSCGR & IHQS$_p$-CG & PDNet & FBPNet & MetaInvNet & IHQS$_p$-Net* & IHQS$_p$-Net \\ \hline
				
				Noise & \multicolumn{10}{c|}{$\sigma=0.3$}\\\hline
				
				180 & 13.18/16.79 & 10.61/14.66 & 4.73/6.74 & 4.51/6.41 & 7.98/453.28 & \textbf{4.48}/\textbf{6.37} & 5.69/8.07 & 4.74/6.67 & 4.17/5.79 & 4.29/5.99 & {\color{red} 4.07}/{\color{red} 5.68} \\ 
				120 & 17.45/23.56 & 11.97/17.45 & 5.51/8.01 & 5.10/7.48 & 8.06/438.85 & \textbf{5.03}/\textbf{7.37} & 6.33/9.22 & 5.34/7.78 & {\color{red} 4.52}/{\color{red} 6.37} & 4.83/6.89 & 4.54/6.42 \\ 
				90 & 19.63/25.21 & 12.24/16.96 & 6.25/9.26 & 5.62/8.50 & 9.12/503.83 & \textbf{5.50}/\textbf{8.29} & 6.57/9.90 & 5.76/8.48 & 4.78/6.82 & 5.28/7.69 & {\color{red} 4.70}/{\color{red} 6.74} \\ 
				60 & 24.54/31.38 & 14.36/19.48 & 7.61/11.62 & 6.57/10.43 & 11.07/619.98 & \textbf{6.36}/\textbf{10.11} & 7.56/11.94 & 6.48/10.12 & {\color{red} 5.37}/{\color{red} 7.79} & 6.08/9.20 & 5.45/7.90 \\ \hline

				Noise & \multicolumn{11}{c|}{$\sigma=0.3$ and $I_0=5\times 10^{5}$}\\\hline
				
				180 & 14.61/18.54 & 10.61/14.67 & 5.45/7.61 & 5.14/\textbf{7.31} & 9.04/491.19 & \textbf{5.13}/7.32 & 6.01/8.56 & 4.88/6.94 & 4.58/6.48 & 4.86/6.83 & {\color{red} 4.47}/{\color{red} 6.30} \\ 
				120 & 19.14/25.41 & 11.98/17.45 & 6.35/9.01 & 5.64/8.29 & 14.26/731.87 & \textbf{5.58}/\textbf{8.22} & 6.84/10.00 & 5.48/7.97 & {\color{red} 4.84}/6.99 & 5.37/7.69 & 4.85/{\color{red} 6.94} \\ 
				90 & 21.53/27.43 & 12.24/16.96 & 7.12/10.28 & 6.16/9.32 & 13.22/692.11 & \textbf{6.02}/\textbf{9.12} & 7.37/11.23 & 6.01/8.90 & 5.19/7.46 & 5.78/8.38 & {\color{red} 5.05}/{\color{red} 7.26} \\ 
				60 & 26.72/33.85 & 14.30/19.43 & 8.44/12.54 & 7.19/11.31 & 13.36/718.63 & \textbf{6.97}/\textbf{11.04} & 8.32/13.56 & 6.43/10.03 & 5.71/8.44 & 6.37/9.50 & {\color{red} 5.47}/{\color{red} 8.05} \\ \hline

				Data & \multicolumn{11}{c|}{Covid-19 Dataset}  \\ \hline
				Method & FBP & BM3D & HQS-CG & HQS$_p$-CG & PWLS-CSCGR & IHQS$_p$-CG & PDNet & FBPNet & MetaInvNet & IHQS$_p$-Net* & IHQS$_p$-Net \\ \hline
				
				Noise & \multicolumn{10}{c|}{$\sigma=0.3$}\\\hline
				
				180 & 12.88/16.26 & 10.49/13.64 & 4.03/5.62 & 3.77/5.27 & 6.86/289.66 & \textbf{3.73}/\textbf{5.19} & 5.50/7.50 & 4.14/5.68 & 3.54/4.78 & 3.65/4.97 & {\color{red} 3.42}/{\color{red} 4.65} \\ 
				120 & 16.75/22.56 & 11.69/15.79 & 4.75/6.84 & 4.33/6.32 & 7.16/296.62 & \textbf{4.21}/\textbf{6.07} & 6.06/8.50 & 4.75/6.76 & {\color{red} 3.77}/{\color{red} 5.19} & 4.19/5.85 & 3.85/5.32 \\ 
				90 & 18.60/24.11 & 12.01/15.91 & 5.51/8.23 & 4.86/7.44 & 7.99/342.59 & \textbf{4.69}/\textbf{7.06} & 6.15/8.99 & 5.15/7.46 & 4.04/5.69 & 4.58/6.57 & {\color{red} 3.94}/{\color{red} 5.52} \\ 
				60 & 23.30/30.45 & 14.40/19.70 & 6.98/11.05 & 5.93/9.81 & 9.81/432.80 & \textbf{5.61}/\textbf{9.19} & 7.08/11.03 & 6.13/9.80 & {\color{red} 4.54}/{\color{red} 6.62} & 5.41/8.06 & 4.57/6.66 \\ 
				
				\hline
				
				Noise & \multicolumn{11}{c|}{$\sigma=0.3$ and $I_0=5\times 10^{5}$}\\ \hline
				
				180 & 14.35/18.14 & 10.49/13.65 & 4.84/6.63 & 4.47/6.29 & 8.07/326.25 & \textbf{4.44}/\textbf{6.26} & 5.87/8.04 & 4.28/5.94 & 3.87/5.40 & 4.32/5.99 & {\color{red} 3.80}/{\color{red} 5.24} \\ 
				120 & 18.59/24.58 & 11.70/15.81 & 5.68/7.98 & 4.92/7.20 & 13.59/544.19 & \textbf{4.82}/\textbf{7.02} & 6.23/8.80 & 4.88/6.94 & 4.18/5.92 & 4.80/6.78 & {\color{red} 4.14}/{\color{red} 5.80} \\ 
				90 & 20.72/26.58 & 12.01/15.93 & 6.45/9.35 & 5.44/8.33 & 12.32/504.26 & \textbf{5.25}/\textbf{7.97} & 7.08/10.47 & 5.26/7.76 & 4.46/6.39 & 5.22/7.47 & {\color{red} 4.30}/{\color{red} 6.09} \\ 
				60 & 25.79/33.21 & 14.37/19.68 & 7.84/12.00 & 6.57/10.68 & 12.22/517.15 & \textbf{6.27}/\textbf{10.16} & 7.72/12.48 & 6.05/9.33 & 5.02/7.45 & 5.88/8.75 & {\color{red} 4.68}/{\color{red} 6.81} \\ 
				\hline
				
		\end{tabular}}
	\end{center}
	\label{rmse}
\end{table*}

\begin{figure*}[!t]
	\begin{center}
		\scriptsize
		\addtolength{\tabcolsep}{-5pt}
		\renewcommand\arraystretch{0.5}
		{\fontsize{9pt}{\baselineskip}\selectfont
			\begin{tabular}{cccccc}
				\includegraphics[width=0.155\linewidth]{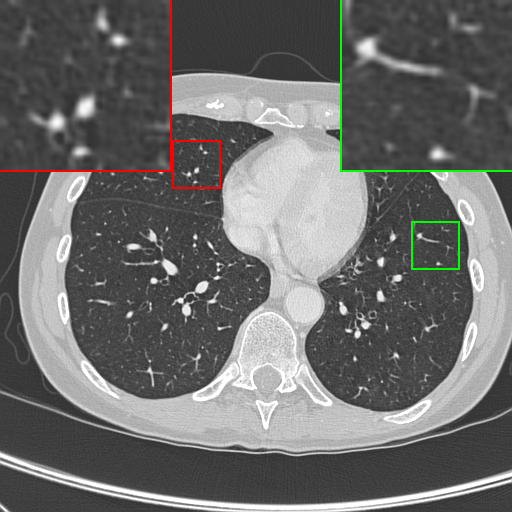} &
				\includegraphics[width=0.155\linewidth]{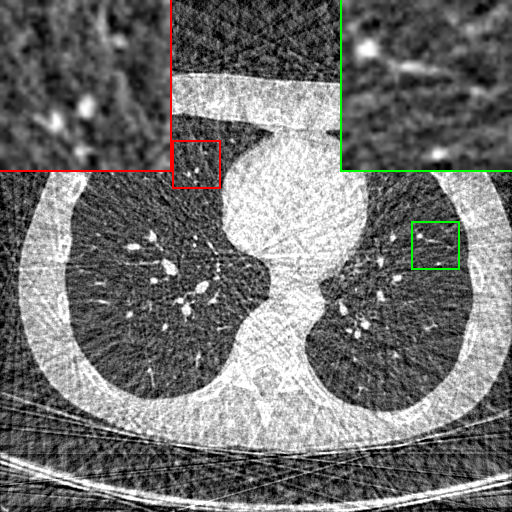} &
				\includegraphics[width=0.155\linewidth]{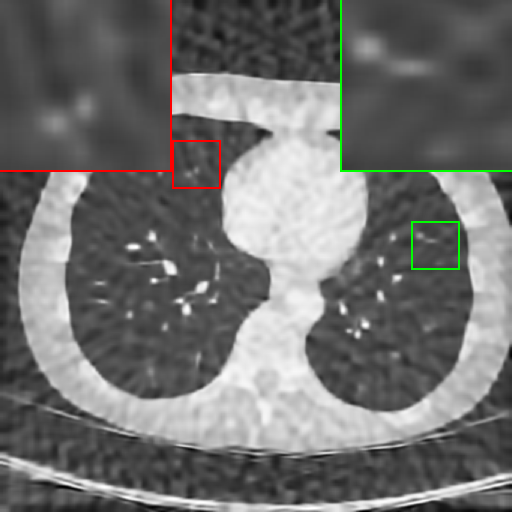} &
				\includegraphics[width=0.155\linewidth]{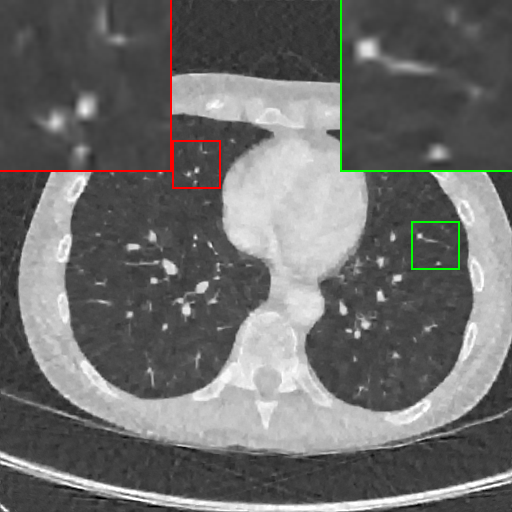} &
				\includegraphics[width=0.155\linewidth]{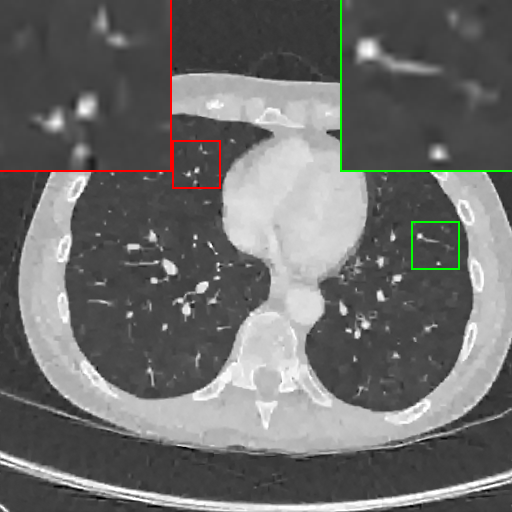} & 
				\includegraphics[width=0.155\linewidth]{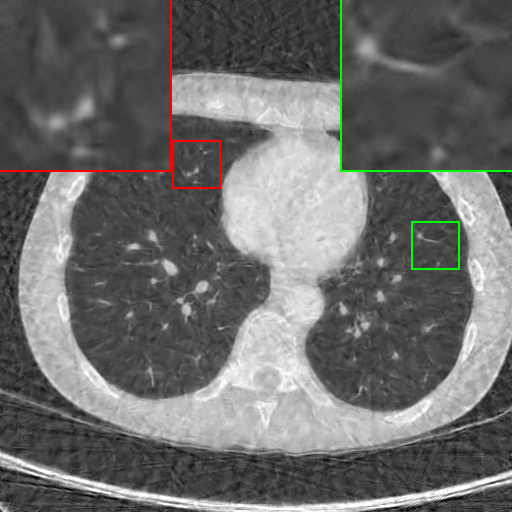} \\
				
				(a) Ground truth & (b) FBP & (c) BM3D & (d) HQS-CG & (e) HQS$_p$-CG & (f) PWLS-CSCGR  \\
				
				\includegraphics[width=0.155\linewidth]{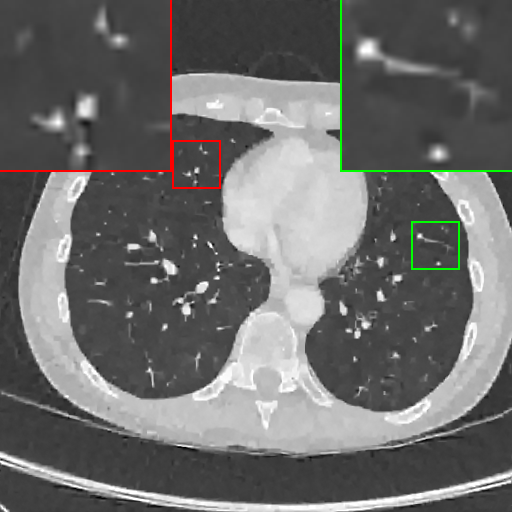} &
				\includegraphics[width=0.155\linewidth]{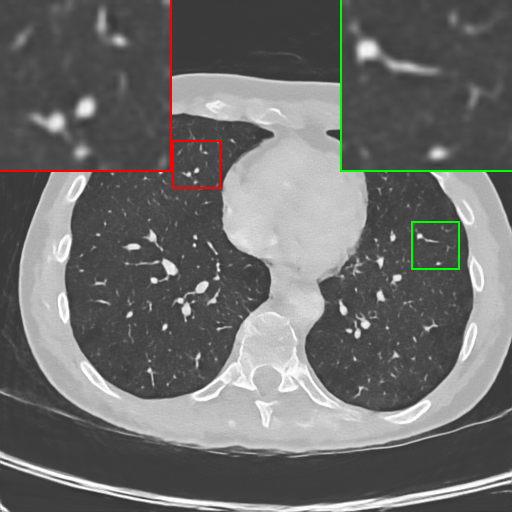} &
				\includegraphics[width=0.155\linewidth]{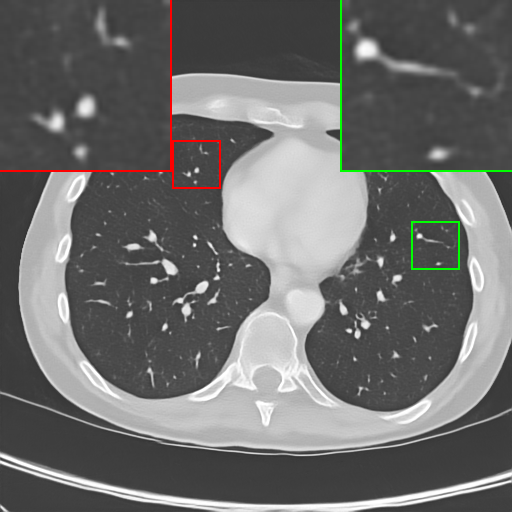} &
				\includegraphics[width=0.155\linewidth]{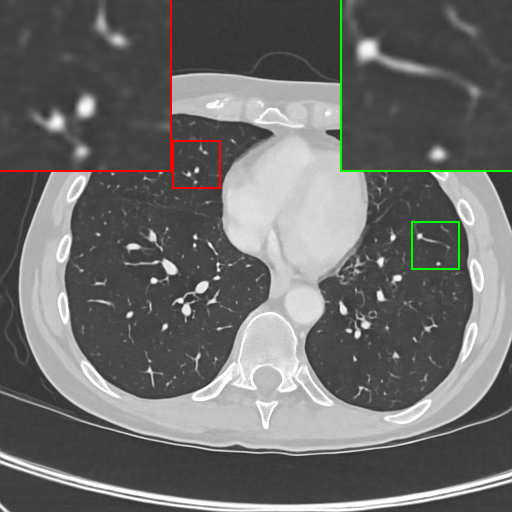} &
				\includegraphics[width=0.155\linewidth]{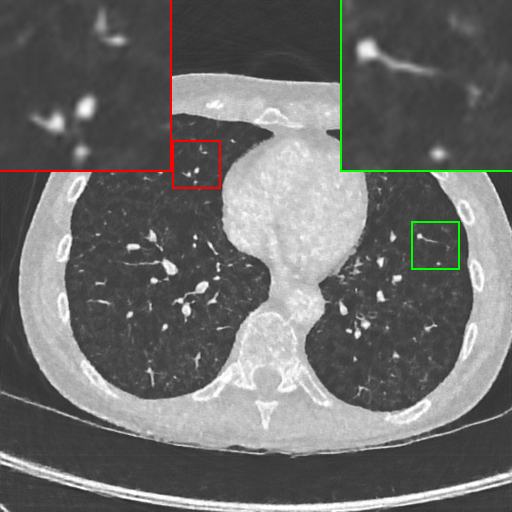} &
				\includegraphics[width=0.155\linewidth]{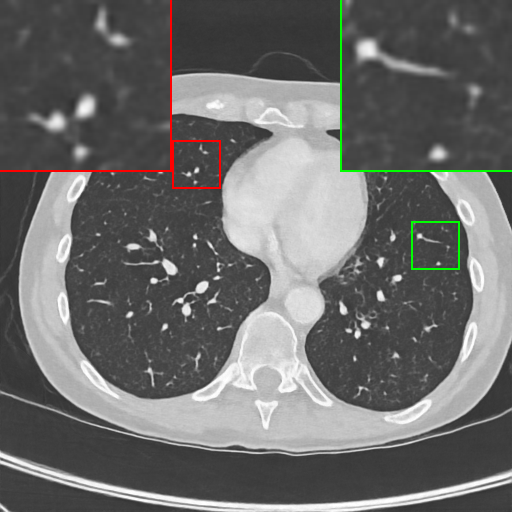} \\
				
				(g) IHQS$_p$-CG & (h) PDNet & (i) FBPNet & (j) MetaInvNet & (k) IHQS$_p$-Net* & (l) IHQS$_p$-Net \\
			\end{tabular}
		}
	\end{center}
	\caption{60 sparse views CT image reconstruction results and magnified ROIs from AAPM. The sinogram is corrupted by $\sigma=0.3$. 
	}
	\label{reconstruction2}
\end{figure*}

\begin{theorem}
	Let $\left\{v^{k}\right\}$ and $\left\{\bar{v}^{k}\right\}$ be the sequences generated by our algorithm. Then any cluster point of $\left\{v^{k}\right\}$ is the global minimum point of $L$. 
\end{theorem}

\begin{proof}[Proof]
	By Lemmal \ref{lemma1}, the sequences $\left\{v^{k}\right\}$ and $\left\{\bar{v}^{k}\right\}$ are bounded, which implies the existence of convergent subsequences, denoted as
	\begin{equation*}
		v^{k_j} \rightarrow v^{*}, j \rightarrow \infty, \quad \mathrm{and} \quad \bar{v}^{k_j} \rightarrow \bar{v}^{*}, j \rightarrow \infty.
	\end{equation*}
	And further, it can be proved from (6) and (8) that 
	\begin{equation}
		\left \{
		\begin{array}{ll}
			\bar{z}^{k+1} + \alpha \bar{z}^{k} = (1+\alpha)z^{k+1}, \\
			\bar{u}^{k+1} + \beta \bar{u}^{k} = (1+\beta)u^{k+1}. \\
		\end{array}
		\right.
		\label{eq14}
	\end{equation}
	Passing to the limit in (\ref{eq14}) along the subsequences $\left\{v^{k_j}\right\}$ and $\left\{\bar{v}^{k_j}\right\}$ yields
	\begin{equation}
		\left \{
		\begin{array}{ll}
			(1+\alpha) \bar{z}^{*} = (1+\alpha)z^{*}, \\
			(1+\beta) \bar{u}^{*} = (1+\beta)u^{*}, \\
		\end{array}
		\right.
	\end{equation}
	which means $\bar{z}^{*} = z^{*}$ and $\bar{u}^{*} = u^{*}$. It follows from (5) and (7) that
	\begin{equation}
		\left \{
		\begin{array}{ll}
			0 \in \lambda\partial\|z^{k+1}\|_{p}^{p} - \sum_{i=1}^{m}\gamma_{i}(W_{i}\bar{u}^{k}-z^{k+1}_{i}), \\
			0 = A^{T}(Au^{k+1}-f)+\sum_{i=1}^{m}\gamma_{i}W_{i}^{T}(W_{i}u^{k+1}-\bar{z}^{k+1}_{i}).\\
		\end{array}
		\right.
		\label{eq15}
	\end{equation}
	Taking limits on both sides of (\ref{eq15}) along the subsequences $\left\{v^{k_j}\right\}$ and $\left\{\bar{v}^{k_j}\right\}$, when $j \rightarrow \infty$ and using closedness of subdifferential, we obtain
	\begin{equation}
		\left \{
		\begin{array}{ll}
			0 \in \lambda\partial\|z^{*}\|_{p}^{p} - \sum_{i=1}^{m}\gamma_{i}(W_{i}{u}^{*}-z^{*}_{i}), \\
			0 = A^{T}(Au^{*}-f)+\sum_{i=1}^{m}\gamma_{i}W_{i}^{T}(W_{i}u^{*}-{z}^{*}_{i}).\\
		\end{array}
		\right.
		\label{eq16}
	\end{equation}
	In particular, $v^{*}$ is a stationary point of $L$. $L$ is a quasi-convex function, since $\|z\|_{p}^{p}$ is a quasi-convex function. Consequently, this property results in that  $v^{*} = (z^{*},u^{*})$ is the global minimum point of $L$.
\end{proof}

\section{Supplementary experimental data}

\subsection{Comparison of MAE and RMSE}
Table \ref{rmse} gives the comparison of MAE and RMSE for all the compared algorithms. The results in Table \ref{rmse} are similar to those in Table I. IHQ$S_p$-CG performs optimally in the traditional algorithm. IHQS$_p$-Net also performs best among the deep learning algorithms in aggregate. This shows the robust results of the proposed IHQS$_p$ algorithm under different metrics.

\subsection{Supplementary visualisation results}
Fig. \ref{reconstruction2} illustrates a qualitative comparison of the 60-degree observation resulting from the reconstruction. It can be noticed that MetaInvNet is prone to pseudo-generation. HQS-CG, on the other hand, is not able to complete the reconstruction details. Overall, IHQS$_p$-CG and IHQS$_p$-Net are good at preserving details in both traditional and deep learning algorithms.

\end{document}